\documentclass[a4paper, UKenglish, cleveref, autoref]{lipics-v2021}

\pdfoutput=1 
\hideLIPIcs  
\nolinenumbers

\title{Smaller Circuits for Bit Addition} 

\author{Mikhail Goncharov}{Neapolis University Pafos and JetBrains Research}{}{}{}

\author{Alexander~S. Kulikov}{JetBrains Research\and \url{https://alexanderskulikov.github.io}}{alexander.s.kulikov@gmail.com}{https://orcid.org/0000-0002-5656-0336}{}

\author{Georgie Levtsov}{Neapolis University Pafos and JetBrains Research}{}{}{}

\authorrunning{M.~Goncharov, A.~S.~Kulikov, and G.~Levtsov} 

\Copyright{Mikhail Goncharov, Alexander~S. Kulikov, and Georgie Levtsov} 



\supplement{GitHub repository: \url{https://github.com/spbsat/cirbo}}

\usepackage{booktabs}
\usepackage{tikz}
\usetikzlibrary{math, fit}
\usepackage{listings}
\usepackage{multirow}

\tikzstyle{dot} = [circle, fill=black, inner sep=0mm, minimum size=1mm]
\tikzstyle{l} = [dotted, thin]
\tikzstyle{input} = [draw=none, inner sep=.2mm]
\tikzstyle{gate} = [draw, circle, inner sep=.2mm, minimum size=2mm]
\tikzstyle{outgate} = [gate, thick]
\tikzstyle{wire} = [draw,->]
\tikzstyle{notwire} = [draw,->,dashed]
\tikzstyle{pair} = [rectangle, draw=gray, inner sep=.7mm]
\tikzmath{\d=0.4;} 

\usepackage[textwidth=15mm]{todonotes}

\DeclareMathOperator{\SUM}{SUM}
\DeclareMathOperator{\ADD}{ADD}
\DeclareMathOperator{\MULT}{MULT}
\DeclareMathOperator{\BA}{BA}
\DeclareMathOperator{\MDFA}{MDFA}
\DeclareMathOperator{\INC}{INC}
\DeclareMathOperator{\size}{size}

\EventEditors{John Q. Open and Joan R. Access}
\EventNoEds{2}
\EventLongTitle{42nd Conference on Very Important Topics (CVIT 2016)}
\EventShortTitle{CVIT 2016}
\EventAcronym{CVIT}
\EventYear{2016}
\EventDate{December 24--27, 2016}
\EventLocation{Little Whinging, United Kingdom}
\EventLogo{}
\SeriesVolume{42}
\ArticleNo{23}

\begin{document}
    \maketitle

    \begin{abstract}
        Bit addition arises virtually everywhere in~digital circuits:
        arithmetic operations,
        increment/decrement operators,
        computing addresses and table indices, and so~on.
        Since bit addition is~such a~basic task in~Boolean circuit synthesis,
        a~lot of~research has been done on~constructing efficient circuits
        for various special cases of~it. A~vast majority of~these results are devoted to~optimizing the circuit \emph{depth} (also known as~delay).

        In~this paper, we~investigate the circuit \emph{size} (also known as~area)
        over the full binary basis of~bit addition. Though most of~the known circuits are built from Half Adders and Full Adders,
        we~show that, in~many interesting scenarios, these circuits have suboptimal size.
        Namely, we~improve an~upper bound $5n-3m$ to~$4.5n-2m$,
        where $n$~is the number of~input bits and $m$~is the number of~output bits.
        In~the regimes where $m$~is small compared to~$n$
        (for example, for computing the sum
        of~$n$~bits or~multiplying two $n$-bit integers),
        this leads to~$10\%$ improvement.

        We~complement our theoretical result by~an~open-source implementation
        of~generators producing circuits for bit addition and multiplication.
        The generators allow one to~produce the corresponding circuits
        in~two lines of~code and to~compare them to~existing designs.
    \end{abstract}


    \section{Overview}
    Bit addition arises virtually everywhere in~digital circuits:
    arithmetic operations,
    increment/decrement operators,
    computing addresses and table indices, and so~on.
    Three specific scenarios where it~is used frequently are listed below.
    \begin{itemize}
        \item Adding two $n$-bit numbers.
        \item Computing a~symmetric Boolean function
        (such as~majority or~sorting).
        A~natural way of~doing this is~to~first compute
        the binary representation of~the sum of~$n$~input bits
        (that~is, to~compress $n$~bits into about $\log_2 n$ bits)
        and then to~compute the function at~hand
        out of~the computed binary representation.
        \item To~multiply two $n$-bit numbers, one may first compute
        all partial products (that~is, products of~the bits of~the
        two input numbers) and then sum~up the resulting bits.
    \end{itemize}
    In~terms of~the dot-notation introduced by~Dadda~\cite{dadda}, the three scenarios discussed above are visualized as~shown in~Figure~\ref{figure:dot1}. In~this notation, one places
    bits of~the same significance on~the same vertical layer.

    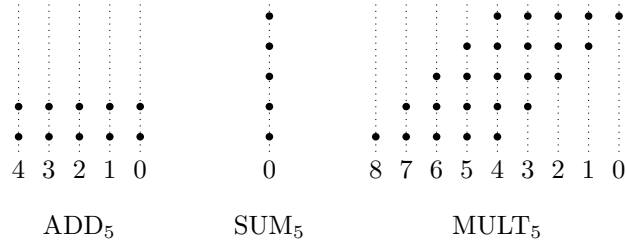
\begin{figure}[ht]
        \begin{center}
            \begin{tikzpicture}
                \begin{scope}[xshift=25mm]
                    \foreach \x [count=\n from 0] in {0} {
                        \draw[l] (\x * \d, -1) -- (\x * \d, 1);
                        \node[below] at (\x * \d, -1) {$\n$};
                    }

                    \foreach \y in {-2, ..., 2}
                    \node[dot] at (0, \y * \d) {};
                    \node at (0, -2) {$\SUM_5$};
                \end{scope}

                \begin{scope}[xshift=0mm]
                    \foreach \x [count=\n from 0] in {2, 1, ..., -2} {
                        \draw[l] (\x * \d, -1) -- (\x * \d, 1);
                        \node[below] at (\x * \d, -1) {$\n$};
                    }

                    \foreach \x in {-2, ..., 2} {
                        \node[dot] at (\x * \d, - 2 * \d) {};
                        \node[dot] at (\x * \d, - \d) {};
                    }
                    \node at (0, -2) {$\ADD_5$};
                \end{scope}

                \begin{scope}[xshift=55mm]
                    \foreach \x [count=\n from 0] in {4, 3, ..., -4} {
                        \draw[l] (\x * \d, -1) -- (\x * \d, 1);
                        \node[below] at (\x * \d, -1) {$\n$};
                    }

                    \foreach \y in {-2, ..., 2}
                    \foreach \x in {-2, ..., 2} {
                        \node[dot] at (\x * \d + \y * \d, \y * \d) {};
                    }
                    \node at (0, -2) {$\MULT_5$};
                \end{scope}
            \end{tikzpicture}
        \end{center}
        \caption{Dot diagrams for three Boolean functions: $\ADD_5$ adds two five-bit numbers, $\SUM_5$ adds five bits, and $\MULT_5$ adds five (appropriately shifted) five-bit numbers.}
        \label{figure:dot1}
    \end{figure}

    There are many other cases where one needs to~add bits.
    Say, one may want
    to~add a~single bit to an~$n$-bit number (the increment operation
    is a~special case),
    or~to~add three $n$-bit numbers, or~to~add a~few bits of~varying significance,
    see Figure~\ref{figure:dot2}.

    \begin{figure}[ht]
        \begin{center}
            \begin{tikzpicture}
                \begin{scope}
                    \foreach \x [count=\n from 0] in {2, 1, ..., -2} {
                        \draw[l] (\x * \d, -1) -- (\x * \d, .5);
                        \node[below] at (\x * \d, -1) {$\n$};
                    }

                    \foreach \x in {-2, ..., 2}
                    \node[dot] at (\x * \d, - 2 * \d) {};
                    \node[dot] at (2 * \d, -\d) {};
                \end{scope}

                \begin{scope}[xshift=30mm]
                    \foreach \x [count=\n from 0] in {2, 1, ..., -2} {
                        \draw[l] (\x * \d, -1) -- (\x * \d, .5);
                        \node[below] at (\x * \d, -1) {$\n$};
                    }

                    \foreach \x in {-2, ..., 2} {
                        \node[dot] at (\x * \d, 0) {};
                        \node[dot] at (\x * \d, -\d) {};
                        \node[dot] at (\x * \d, -2 * \d) {};
                    }
                \end{scope}

                \begin{scope}[xshift=60mm]
                    \foreach \x [count=\n from 0] in {2, 1, ..., -2} {
                        \draw[l] (\x * \d, -1) -- (\x * \d, .5);
                        \node[below] at (\x * \d, -1) {$\n$};
                    }

                    \foreach \x/\y in {-2/-1, -1/-1, -1/0, -1/1, 0/-1, 1/-1, 1/0, 2/-1}
                    \node[dot] at (\x * \d, \y * \d - \d) {};
                \end{scope}
            \end{tikzpicture}
        \end{center}
        \caption{More scenarios of~adding bits of~varying significance.}
        \label{figure:dot2}
    \end{figure}
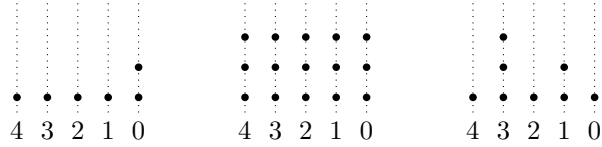

    A~function capturing all such scenarios is~known as~\emph{bit adder}
    \[\BA_n^{s_1, \dotsc, s_n} \colon \{0,1\}^n \to \{0,1\}^m.\] It~is parameterized by~the \emph{significance vector}
    $s=(s_1, \dotsc, s_n) \in \mathbb{Z}_{\ge 0}^n$, takes $n$~input bits $(x_1, \dotsc, x_n) \in \{0,1\}^n$, and outputs
    the binary representation~of
    \[\sum_{i=1}^{n}2^{s_i}x_i.\]
    This way, $\SUM_n=\BA^{0,0,\dotsc,0}_n$ and $\ADD_n=\BA^{0,0,1,1,\dotsc,n-1,n-1}_{2n}$.

    Since bit addition is~such a~basic task in~Boolean circuit synthesis,
    a~lot of~research has been done on~constructing efficient circuits
    for various special cases of~it, see, for example,
    \cite{DBLP:journals/cc/PatersonZ93,
        DBLP:conf/arith/MartelORS95,
        DBLP:journals/tc/StellingMOR98,
        DBLP:conf/arith/BickerstaffSS01}.
    A~vast majority of~these results is~devoted to~optimizing the circuit \emph{depth} (also known as~delay).
    In~this paper, we~investigate the circuit \emph{size} (also known as~area) of~bit addition. Specifically, we~study circuits over the full binary basis.

    Two basic building blocks for adding bits are known as~Half Adder~(HA)
    and Full Adder~(FA). They compute the binary representation of~the sum
    of~two and three bits, respectively (that~is, $\SUM_2$ and $\SUM_3$).
    In~the full binary basis, they can be~implemented in~two and five gates, respectively, see Figure~\ref{figure:sum23}.

    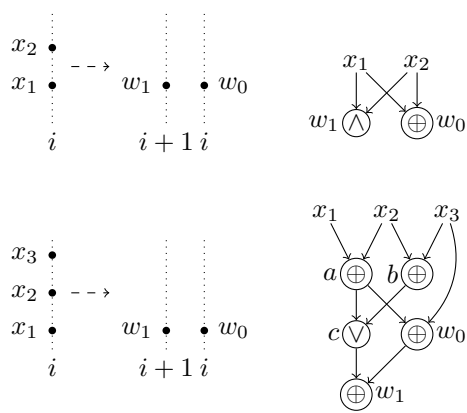
\begin{figure}
        \begin{center}
            \begin{tikzpicture}
                \begin{scope}[label distance=-.9mm, scale=.8]
                    \foreach \n/\x/\y in {1/0/1, 2/1/1}
                    \node[input] (x\n) at (\x, \y) {$x_{\n}$};
                    \node[gate, label=left:$w_1$] (g1) at (0,0) {$\land$};
                    \node[gate, label=right:$w_0$] (g2) at (1,0) {$\oplus$};
                    \foreach \f/\t in {x1/g1, x1/g2, x2/g1, x2/g2}
                    \draw[->] (\f) -- (\t);

                    \begin{scope}[yshift=-45mm, xshift=-5mm]
                        \foreach \n/\x/\y in {1/0/3, 2/1/3, 3/2/3}
                        \node[input] (x\n) at (\x, \y) {$x_{\n}$};
                        \node[gate,label=left:$a$] (g1) at (0.5,2) {$\oplus$};
                        \node[gate,label=left:$b$] (g2) at (1.5,2) {$\oplus$};
                        \node[gate,label=left:$c$] (g3) at (0.5,1) {$\lor$};
                        \node[gate, label=right:$w_0$] (g4) at (1.5,1) {$\oplus$};
                        \node[gate, label=right:$w_1$] (g5) at (0.5,0) {$\oplus$};
                        \foreach \f/\t in {x1/g1, x2/g1, x2/g2, x3/g2, g1/g3, g2/g3, g1/g4, g3/g5, g4/g5}
                        \draw[->] (\f) -- (\t);
                        \path (x3) edge[bend left,->] (g4);
                    \end{scope}
                \end{scope}

                \begin{scope}[xshift=-40mm, yshift=10mm]
                    \draw[l] (0, -1) -- (0, 0.5); \node[below] at (0, -1) {$i$};
                    \node[dot, label=left:$x_1$] at (0, -0.5) {};
                    \node[dot, label=left:$x_2$] at (0, 0) {};

                    \draw[dashed, ->] (0.25, -.25) -- (0.75, -.25);

                    \draw[l] (2, -1) -- (2, 0.5); \node[below] at (2, -1) {$i$};
                    \draw[l] (1.5, -1) -- (1.5, .5); \node[below] at (1.5, -1) {$i+1$};
                    \node[dot, label=left:$w_1$] at (1.5, -0.5) {};
                    \node[dot, label=right:$w_0$] at (2, -0.5) {};
                \end{scope}

                \begin{scope}[xshift=-40mm, yshift=-20mm]
                    \draw[l] (0, -1) -- (0, 0.5); \node[below] at (0, -1) {$i$};
                    \node[dot, label=left:$x_1$] at (0, -0.75) {};
                    \node[dot, label=left:$x_2$] at (0, -0.25) {};
                    \node[dot, label=left:$x_3$] at (0, 0.25) {};

                    \draw[dashed, ->] (0.25, -.25) -- (0.75, -.25);

                    \draw[l] (2, -1) -- (2, 0.5); \node[below] at (2, -1) {$i$};
                    \draw[l] (1.5, -1) -- (1.5, .5); \node[below] at (1.5, -1) {$i+1$};
                    \node[dot, label=left:$w_1$] at (1.5, -0.75) {};
                    \node[dot, label=right:$w_0$] at (2, -0.75) {};
                \end{scope}
            \end{tikzpicture}
        \end{center}
        \caption{The Half Adder (top) and Full Adder (bottom): dot diagrams and circuits.}
        \label{figure:sum23}
    \end{figure}

    Using Half Adders and Full Adders, one can synthesize a~bit adder using the following algorithm that goes back to~Napier's \emph{Rabdologiæ} (1617),
    as~modernized by~Dadda~\cite{dadda}.
    \begin{quote}
        Process the bits layer by~layer, in~the order of~increasing significance.
        While the current significance layer~$i$ contains at~least three bits,
        take three of~them and apply the Full Adder to~replace them with a~pair of~bits
        of~significance $i$~and~$i+1$. If~there are two bits left at~the current layer~$i$, apply the Half Adder to~them to~get a~pair of~bits of~significance $i$~and~$i+1$.
    \end{quote}
    This algorithm ensures that, for any vector $s \in \mathbb{Z}_{\ge 0}^n$,
    \[\operatorname{size}(\BA_n^s) \le 5n-3m.\]
    Indeed, each application of the Full Adder reduces the number of~bits by~one,
    hence the total cost of~all Full Adders is~at~most $5(n-m)$. The Half Adder is~applied at~most once for each of~the significance layers, hence
    the total cost of~all Half Adders is~at~most~$2m$. Hence, the total size
    is~at~most $5(n-m)+2m=5n-3m$.

    By~applying this algorithm to~partial products of~bits of~two input $n$-bit numbers, one gets the well-known Dadda multiplier circuit~\cite{dadda}.
    For many vectors~$s$, the upper bound
    $5n-3m$ is~loose:
    it~does not match the size of~the actual circuit
    produced by~the algorithm.
    A~straightforward example is $s=(0,1,\dotsc,n-1)$:
    in~this case, no~gates are needed whereas the upper bound is~$2n$.
    It~is also worth noting that, in~some cases, the resulting circuit
    is~\emph{provably} optimal.
    For example, for the $\ADD_n$ function (that computes the sum of~two $n$-bit integers),
    the method constructs a~circuit out of a~single Half Adder and $(n-1)$
    Full Adders. The resulting circuit is~known as~\emph{ripple-carry adder} and has size $5n-3$.
    Red'kin~\cite{Red81} proved that there~is no~smaller circuit
    for this function.

    At~the same time, in~many scenarios,
    not only the bound $5n-3m$ is~loose,
    but also the circuit produced by~the algorithm
    is~suboptimal.
    For example, for $\SUM_5$, it~gives a~circuit of~size~$12$ consisting
    of~two Full Adders and one Half Adder, see Figure~\ref{figure:sumfive}.
    However, $\SUM_5$
    can be~computed by a~circuit of~size~$11$ as~shown by~\cite{DBLP:conf/mfcs/KulikovPS22} (see also Figure~\ref{figure:mdfa} later in~the text).
    In~general, whereas the algorithm produces a~circuit of~size about~$5n$
    for $\SUM_n$, this function can be~computed by~a~circuit of~size about $4.5n$
    as~shown by~Demenkov et~al.~\cite{DBLP:journals/ipl/DemenkovKKY10}.

    \begin{figure}[ht]
        \begin{center}
            \begin{tikzpicture}
                \begin{scope}[scale=1]
                    \begin{scope}[scale=1.2, yshift=10mm]
                        \begin{scope}[yshift=20mm]
                            \foreach \x [count=\n from 0] in {0} {
                                \draw[l] (\x * \d, -1) -- (\x * \d, 1.5);
                                \node[below] at (\x * \d, -1) {$\n$};
                            }
                            \foreach \n in {-2,-1,...,2}
                            \node[dot] at (0, \n * \d) {};

                            \path (0.2, 0) edge[->, dashed] node[below] {FA} (0.8, 0);
                        \end{scope}

                        \begin{scope}[xshift=14mm, yshift=20mm]
                            \foreach \x [count=\n from 0] in {0, -1} {
                                \draw[l] (\x * \d, -1) -- (\x * \d, 1.5);
                                \node[below] at (\x * \d, -1) {$\n$};
                            }
                            \foreach \x/\y in {0/-2, 0/-1, 0/0, -1/-2}
                            \node[dot] at (\x * \d, \y * \d) {};
                            \path (0.2, 0) edge[->, dashed] node[below] {FA} (0.8, 0);
                        \end{scope}

                        \begin{scope}[xshift=28mm, yshift=20mm]
                            \foreach \x [count=\n from 0] in {0, -1} {
                                \draw[l] (\x * \d, -1) -- (\x * \d, 1.5);
                                \node[below] at (\x * \d, -1) {$\n$};
                            }
                            \foreach \x/\y in {0/-2, -1/-2, -1/-1}
                            \node[dot] at (\x * \d, \y * \d) {};
                            \path (0.2, 0) edge[->, dashed] node[below] {HA} (0.8, 0);
                        \end{scope}

                        \begin{scope}[xshift=46mm, yshift=20mm]
                            \foreach \x [count=\n from 0] in {0, -1, -2} {
                                \draw[l] (\x * \d, -1) -- (\x * \d, 1.5);
                                \node[below] at (\x * \d, -1) {$\n$};
                                \node[dot] at (\x * \d, -2 * \d) {};
                            }
                        \end{scope}
                    \end{scope}

                    \begin{scope}[yshift=-30mm, scale=.9]
                        \foreach \n/\x/\y in {x_1/0/3, x_2/1/4, x_3/2/4, x_4/4/4, x_5/5/4, w_0/6/3, w_1/6/1.5, w_2/3/0.5}
                        \node[input] (\n) at (\x, \y) {$\n$};
                        \draw (0.5, 2.5) rectangle (2.5, 3.5); \node at (1.5, 3) {FA};
                        \draw (3.5, 2.5) rectangle (5.5, 3.5); \node at (4.5, 3) {FA};
                        \draw (0.5, 1) rectangle (5.5, 2); \node at (3, 1.5) {HA};
                        \foreach \f/\t in {x_1/{0.5, 3}, x_2/{1, 3.5}, x_3/{2, 3.5}, x_4/{4, 3.5}, x_5/{5, 3.5}, {5.5, 3}/w_0, {2.5, 3}/{3.5, 3},
                            {1.5, 2.5}/{1.5, 2}, {4.5, 2.5}/{4.5, 2}, {5.5, 1.5}/w_1, {3, 1}/w_2}
                        \draw[->] (\f) -- (\t);
                    \end{scope}

                    \begin{scope}[label distance=-1mm, xshift=70mm, yshift=20mm]
                        \foreach \n/\x/\y in {1/0/3, 2/1/3, 3/2/3, 4/2.5/1, 5/3.5/1}
                        \node[input] (x\n) at (\x, \y) {$x_{\n}$};
                        \node[gate,label=left:$g_1$] (g1) at (0.5,2) {$\oplus$};
                        \node[gate,label=left:$g_2$] (g2) at (1.5,2) {$\oplus$};
                        \node[gate,label=left:$g_3$] (g3) at (0.5,1) {$\lor$};
                        \node[gate,label=left:$g_4$] (g4) at (1.5,1) {$\oplus$};
                        \node[gate,label=left:$g_5$] (g5) at (0.5,0) {$\oplus$};
                        \node[gate,label=left:$g_6$] (g6) at (2,-1) {$\oplus$};
                        \node[gate,label=right:$g_7$] (g7) at (3,-1) {$\oplus$};
                        \node[gate,label=right:$g_8$] (g8) at (2,-2) {$\lor$};
                        \node[gate, label=right:$w_0$] (g9) at (3,-2) {$\oplus$};
                        \node[gate, label=right:$g_9$] (g10) at (2,-3) {$\oplus$};
                        \node[gate, label=right:$w_1$] (g11) at (2,-4) {$\oplus$};
                        \node[gate, label=left:$w_2$] (g12) at (1,-4) {$\land$};

                        \foreach \f/\t in {x1/g1, x2/g1, x2/g2, x3/g2, g1/g3, g2/g3, g1/g4, g3/g5, g4/g5, g4/g6, x4/g6, x4/g7, x5/g7, g6/g8, g7/g8, g8/g10, g6/g9, g9/g10, g10/g11, g10/g12}
                        \draw[->] (\f) -- (\t);

                        \path (x3) edge[->,bend left] (g4);
                        \path (x5) edge[->,bend left=35] (g9);
                        \path (g5) edge[->,bend right=25] (g11);
                        \path (g5) edge[->,bend right=15] (g12);

                        \draw[dashed] (-0.5,-0.25) rectangle (2,2.5);
                        \draw[dashed] (1.25,-3.25) rectangle (4,-0.5);
                        \draw[dashed] (0,-3.5) rectangle (3,-4.5);
                    \end{scope}
                \end{scope}
            \end{tikzpicture}
        \end{center}
        \caption{A~circuit of~size~$12$ computing~$\SUM_5$ composed of~two Full Adders and one Half Adder: dot notation (top left), block structure (bottom left), and a~circuit (right).}
        \label{figure:sumfive}
    \end{figure}
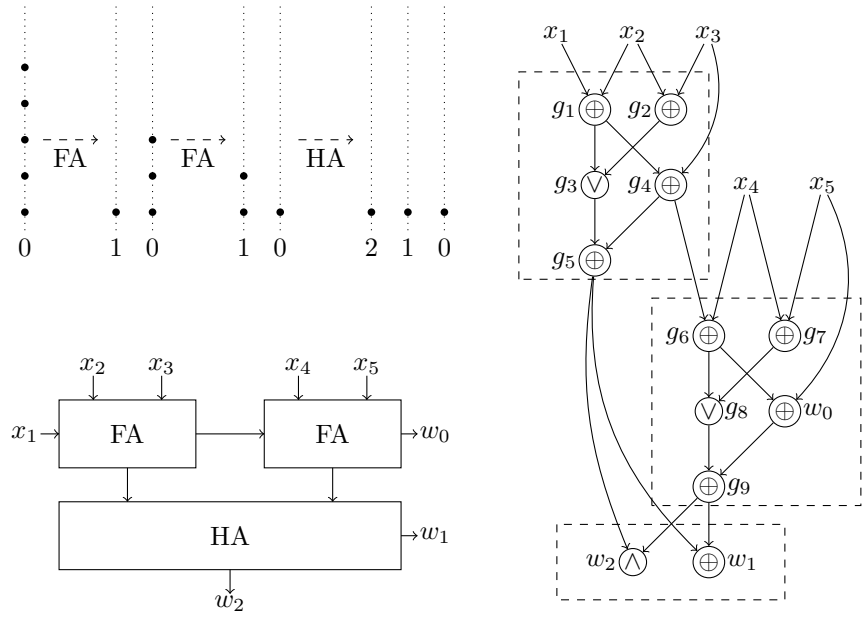

    In~this paper, we~generalize the construction by~Demenkov et~al.
    Namely, we~prove an~upper bound $4.5n-2m$
    for the circuit size of~bit addition.
    In~the regimes where $m$~is~small
    compared to~$n$, this gives a~circuit that is~about $10\%$ smaller.
    This applies to~the Dadda multiplier.
    We~complement our theoretical result by~an~open source implementation
    of~generators producing circuits for bit addition and multiplication.

    \section{General Setting}
    In~this section,
    we~formally introduce the Boolean functions
    studied in~this paper as~well as~the main building blocks
    for computing them.

    \subsection{Boolean Functions}
    The main Boolean function studied in~this paper
    is~\emph{bit adder}
    \[\BA_n^{s_1, \dotsc, s_n} \colon \{0,1\}^n \to \{0,1\}^m.\]
    It~computes the binary representation of~the weighted sum of~input bits:
    \[\sum_{i=1}^{n}2^{s_i}x_i.\]
    In~most interesting scenarios, all bits of~the binary representation of~this sum depend on~the input and the number of~outputs can
    be~expressed as~follows:
    \[m=\left\lceil \log_2\left( \sum_{i=1}^{n}2^{s_i} + 1\right) \right\rceil.\]
    In~such cases,
    \[\BA(x_1, \dotsc, x_n)=(y_0, \dotsc, y_{m-1}) \colon \sum_{i=1}^{n}2^{s_i}x_i=\sum_{i=0}^{m-1}2^iy_i.\]
    However, for some other significance vectors, some of~the bits
    of~the binary representation of~the sum are identically equal to~zero (and thus, do~not depend on~the input). We~exclude such bits from the outputs. Thus, more generally, when we~say that
    \[\BA(x_1, \dotsc, x_n)=(y_0, \dotsc, y_{m-1}),\]
    we~mean that there exists a~vector $t=(t_0, \dotsc, t_{m-1}) \in \mathbb{Z}_{\ge 0}$ such that $t_0 < t_1 < \dotsb < t_{m-1}$ and
    \[\sum_{i=1}^{n}2^{s_i}x_i=\sum_{i=0}^{m-1}2^{t_i}y_i.\]
    It~is not difficult to~see that the vector~$t$ is~unique and that $m \le n$.

    This way, the goal of~bit addition is~to~``flatten''
    the distribution of~bits, that~is, to~leave at~most one bit
    at~each significance layer. Figure~\ref{figure:baexample}
    gives an~example.

    \begin{figure}[ht]
        \begin{center}
            \begin{tikzpicture}
                \begin{scope}
                    \foreach \x [count=\n from 0] in {2, 1, ..., -4} {
                            \draw[l] (\x * \d, -1) -- (\x * \d, .5);
                            \node[below] at (\x * \d, -1) {$\n$};
                        }

                    \foreach \x/\y in {-4/-1, -3/-1, -3/0, -3/1, 1/-1, 1/0, 2/-1}
                    \node[dot] at (\x * \d, \y * \d - \d) {};
                \end{scope}

                \draw[dashed, ->] (1.5, -\d) -- (2.5, -\d);

                \begin{scope}[xshift=40mm]
                    \foreach \x [count=\n from 0] in {5, 4, ..., -2} {
                        \draw[l] (\x * \d, -1) -- (\x * \d, .5);
                        \node[below] at (\x * \d, -1) {$\n$};
                    }
                    \foreach \x in {5, 4, 3, 0, -1, -2}
                        \node[dot] at (\x * \d, -2 * \d) {};
                \end{scope}
            \end{tikzpicture}
        \end{center}
    \caption{The function $\BA_7^{0, 1, 1, 5, 5, 5, 6} \colon \{0,1\}^7 \to \{0,1\}^6$ replaces seven bits of~significance $(0, 1, 1, 5, 5, 5, 6)$ with six bits of~significance $(0, 1, 2, 5, 6, 7)$.}
    \label{figure:baexample}
    \end{figure}
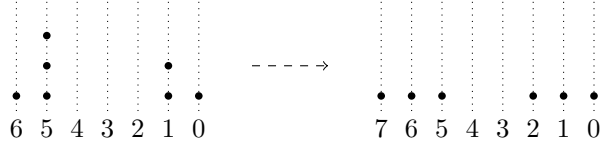

    Many practically important Boolean functions can be~computed using bit summation.
    \begin{itemize}
        \item The function $\SUM_n \colon \{0,1\}^n \to \{0,1\}^{\lceil \log_2(n+1) \rceil}$
        computes the sum of~$n$ bits: \[\SUM_n(x_1, \dotsc, x_n)=\ADD_n^{0,0,\dotsc,0}(x_1, \dotsc, x_n).\]
        \item The function $\ADD_n \colon \{0,1\}^{2n} \to \{0,1\}^{n+1}$ computes the sum
        of~two $n$-bit numbers:
        \[
            \ADD_n(x_0, \dotsc, x_{n-1}, y_0, \dotsc, y_{n-1})
            =\BA_{2n}^{0,\dotsc,n-1,0,\dotsc,n-1}(x_0, \dotsc, x_{n-1}, y_0, \dotsc, y_{n-1}).
        \]
        \item The function $\MULT_n \colon \{0,1\}^{2n} \to \{0,1\}^{2n}$ computes the product
        of~two $n$-bit numbers:
        \[
            \MULT_n(x_0, \dotsc, x_{n-1}, y_0, \dotsc, y_{n-1})=\BA_{n^2}^{(i+j)_{0 \le i, j < n}}\left(\left(x_i \land y_j\right)_{0 \le i, j < n}\right).
        \]
    \end{itemize}

    \subsection{Boolean Circuits}
    A~circuit is~a~natural way of~computing Boolean functions.
    It~is an~acyclic directed graph of in-degree $0$~and~$2$ whose $n+2$~source
    nodes are labeled with input variables
    $x_1, \dotsc, x_n$ and constants $0$~and~$1$, whereas all other nodes
    are labeled with binary Boolean operations.
    The inputs nodes are called input gates, all other nodes are called internal gates.
    Each gate computes
    a~(single-output) Boolean function of~$x_1, \dotsc, x_n$. If~$m$ gates of the
    circuit are marked as~outputs, it~computes a~function of~the form $\{0,1\}^n \to \{0,1\}^m$.
    For a~circuit~$C$, its size, $\operatorname{size}(C)$,
    is~the number of~internal gates of~$C$,
    whereas its depth, $\operatorname{depth}(C)$,
    is~the maximum length of~a~path
    from an~input gate of~$C$ to~an~output gate of~$C$.

    \subsection{Basic Building Blocks}
    As~discussed before, the Half Adder and Full Adder are basic building
    blocks for computing bit addition. Figure~\ref{figure:sum16fa}
    shows how to~synthesize a~circuit of~size~$63$ computing $\SUM_{16}$
    out of~four Half Adders and eleven Full Adders.
    It~is not difficult to~see that a~similar block structure can
    be~used for any~$n$ yielding a~circuit of size at~most~$5n$ for $\SUM_n$.

    \begin{figure}[ht]
        \begin{center}
            \begin{tikzpicture}[scale=0.8]
                \foreach \n in {2,...,17} {
                    \tikzmath{\j=int(\n - 1);}
                    \node[input] (\n) at (\n, 6) {$x_{\j}$};
                }
                \foreach \n in {2, 4, ..., 16} {
                    \draw (\n - 0.25, 4.5) rectangle (\n + 1.25, 5.5);
                    \node at (\n + 0.5, 5) {\ifnumless{\n}{3}{HA}{FA}};
                    \tikzmath{\i=int(\n + 1);}
                    \draw[->] (\n) -- (\n, 5.5);
                    \draw[->] (\i) -- (\i, 5.5);
                    \ifnumless{\n}{16}{\draw[->] (\n + 1.25, 5) -- (\n + 1.75, 5);}{}
                }
                \foreach \n in {2, 6, 10, 14} {
                    \draw (\n - 0.25, 3) rectangle (\n + 3.25, 4);
                    \node at (\n + 1.5, 3.5) {\ifnumless{\n}{3}{HA}{FA}};
                    \draw[->] (\n + 0.5, 4.5) -- (\n + 0.5, 4);
                    \draw[->] (\n + 2.5, 4.5) -- (\n + 2.5, 4);
                    \ifnumless{\n}{14}{\draw[->] (\n + 3.25, 3.5) -- (\n + 3.75, 3.5);}{}
                }
                \draw (1.75, 1.5) rectangle (9.25, 2.5); \node at (5.5, 2) {HA};
                \draw (9.75, 1.5) rectangle (17.25, 2.5); \node at (13.5, 2) {FA};
                \draw (1.75, 0) rectangle (17.25, 1); \node at (9.5, 0.5) {HA};

                \foreach \n/\x/\y in {w_0/18/5, w_1/18/3.5, w_2/18/2, w_3/18/0.5, w_4/9.5/-0.75}
                \node[input] (\n) at (\x, \y) {$\n$};

                \foreach \f/\t in {{17.25, 5}/w_0, {17.25, 3.5}/w_1, {17.25, 2}/w_2, {17.25, 0.5}/w_3, {3.5, 3}/{3.5, 2.5},
                    {7.5, 3}/{7.5, 2.5}, {11.5, 3}/{11.5, 2.5}, {15.5, 3}/{15.5, 2.5},
                    {5.5, 1.5}/{5.5, 1}, {13.5, 1.5}/{13.5, 1}, {9.5, 0}/w_4, {9.25, 2}/{9.75, 2}}
                \draw[->] (\f) -- (\t);
            \end{tikzpicture}
        \end{center}
        \caption{A~circuit computing $\SUM_{16}$ composed out~of four Half Adders and eleven Full Adders. Its size is $4 \cdot 2 + 11 \cdot 5=63$.}
        \label{figure:sum16fa}
    \end{figure}
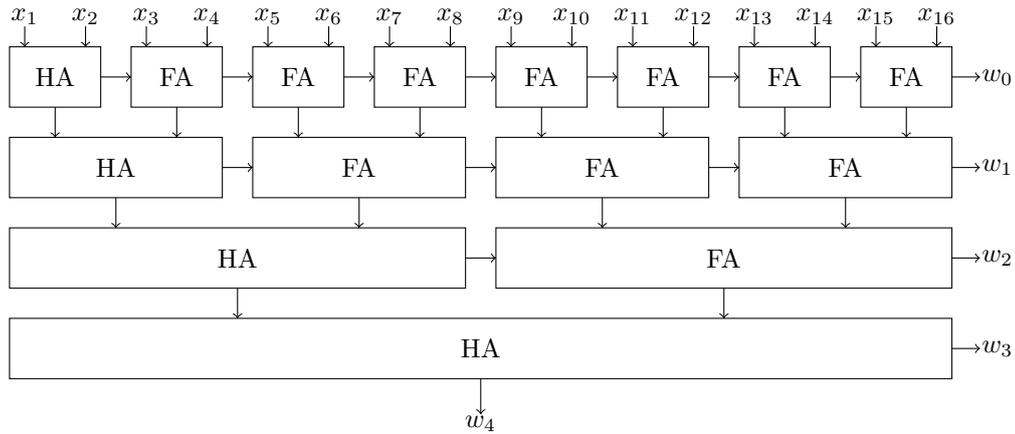

    It~turns out that better circuit designs are possible for $\SUM_n$
    as~shown by~Demenkov et~al.~\cite{DBLP:journals/ipl/DemenkovKKY10}.
    Consider two consecutive Full Adders shown on~the top left of~Figure~\ref{figure:mdfa}. The corresponding function DFA (for Double Full Adder) has the following specification: \[\operatorname{DFA}(x_1, x_2,x_3,x_4,x_5)=(b_0,b_1,a_1) \colon x_1+x_2+x_3+x_4+x_5=b_0+2(b_1+a_1).\]
    Then, MDFA (for Modified Double Full Adder) has the following specification:
    \[\MDFA(x_1 \oplus x_2, x_2, x_3, x_4, x_4 \oplus x_5)=(b_0, a_1, a_1 \oplus b_1).\]
    That~is, for pairs of~bits $(x_1, x_2)$, $(x_4, x_5)$, and $(a_1, b_1)$
    it~uses a~slightly different encoding: $(p, p \oplus q)$ instead of~$(p,q)$.
    We~call such bits paired and show them in~gray boxes in~dot diagrams.
    It~allows one to~compute MDFA in~eight gates (whereas the circuit size of~DFA is~10). Moreover, the corresponding circuit of~size~eight is~just a~part
    of~an~optimal circuit of~size~$11$ computing~$\SUM_5$ shown on~the right
    of~Figure~\ref{figure:mdfa}.

    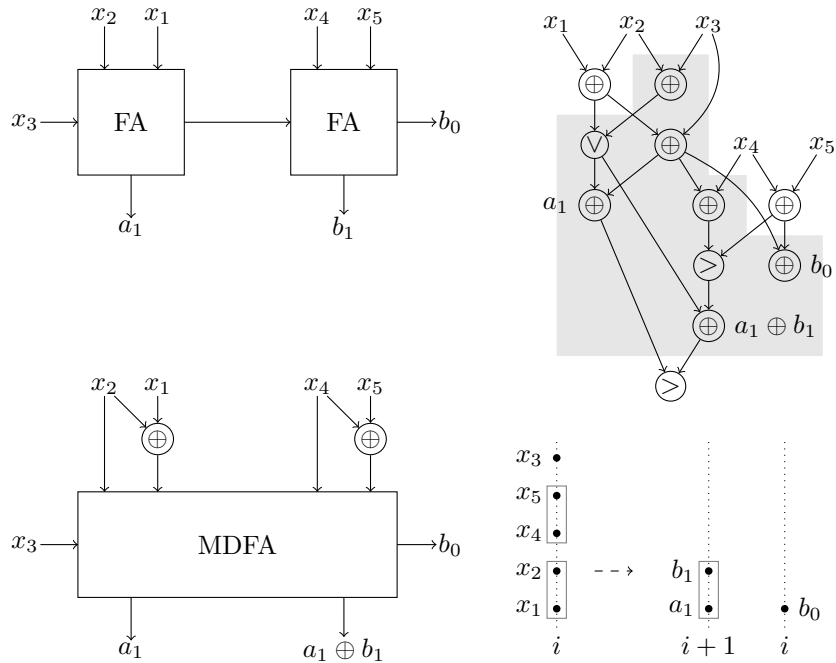
\begin{figure}[ht]
        \begin{center}
            \begin{tikzpicture}[scale=1]
                \begin{scope}[scale=.7]
                    \draw (1,0) rectangle (3,2); \node at (2,1) {FA};
                    \draw (5,0) rectangle (7,2); \node at (6,1) {FA};
                    \foreach \n/\x/\y in {3/0/1, 2/1.5/3, 1/2.5/3, 4/5.5/3, 5/6.5/3}
                    \node[input] (\n) at (\x,\y) {$x_{\n}$};
                    \foreach \n/\t/\x/\y in {a1/a_1/2/-1, b1/b_1/6/-1, b0/b_0/8/1}
                    \node[input] (\n) at (\x,\y) {$\t$};
                    \draw[->] (3)--(1,1);
                    \draw[->] (2)--(1.5,2);
                    \draw[->] (1)--(2.5,2);
                    \draw[->] (4)--(5.5,2);
                    \draw[->] (5)--(6.5,2);
                    \draw[->] (3,1)--(5,1);
                    \draw[->] (7,1)--(b0);
                    \draw[->] (2,0)--(a1);
                    \draw[->] (6,0)--(b1);
                \end{scope}

                \begin{scope}[scale=.7,yshift=-80mm]
                    \draw (1,0) rectangle (7,2); \node at (4,1) {MDFA};
                    \foreach \n/\x/\y in {3/0/1, 2/1.5/4, 1/2.5/4, 4/5.5/4, 5/6.5/4}
                    \node[input] (\n) at (\x,\y) {$x_{\n}$};
                    \node[gate] (xor1) at (2.5,3) {$\oplus$};
                    \node[gate] (xor2) at (6.5,3) {$\oplus$};
                    \foreach \n/\t/\x/\y in {a1/a_1/2/-1, b1/{a_1 \oplus b_1}/6/-1, b0/b_0/8/1}
                    \node[input] (\n) at (\x,\y) {$\t$};
                    \draw[->] (3)--(1,1);
                    \draw[->] (2)--(1.5,2);
                    \draw[->] (1) -- (xor1); \draw[->] (2) -- (xor1);
                    \draw[->] (xor1) -- (2.5,2);
                    \draw[->] (4)--(5.5,2);
                    \draw[->] (5)-- (xor2); \draw[->] (xor2) -- (6.5,2); \draw[->] (4) -- (xor2);
                    \draw[->] (7,1)--(b0);
                    \draw[->] (2,0)--(a1);
                    \draw[->] (6,0)--(b1);
                \end{scope}

                \begin{scope}[xshift=75mm, yshift=-50mm]
                    \draw[l] (-0.5, -1) -- (-0.5, 1.5);
                    \node[below] at (-0.5, -1) {$i$};
                    \node[dot, label=left:$x_1$] (x1) at (-0.5, -0.75) {};
                    \node[dot, label=left:$x_2$] (x2) at (-0.5, -0.25) {};
                    \node[dot, label=left:$x_4$] (x4) at (-0.5, 0.25) {};
                    \node[dot, label=left:$x_5$] (x5) at (-0.5, 0.75) {};
                    \node[dot, label=left:$x_3$] (x3) at (-0.5, 1.25) {};
                    \node[fit=(x1) (x2), pair] {};
                    \node[fit=(x4) (x5), pair] {};

                    \draw[dashed, ->] (0, -.25) -- (0.5, -.25);

                    \draw[l] (2.5, -1) -- (2.5, 1.5); \node[below] at (2.5, -1) {$i$};
                    \draw[l] (1.5, -1) -- (1.5, 1.5); \node[below] at (1.5, -1) {$i+1$};
                    \node[dot, label=left:$a_1$] (a1) at (1.5, -0.75) {};
                    \node[dot, label=left:$b_1$] (b1) at (1.5, -0.25) {};
                    \node[dot, label=right:$b_0$] at (2.5, -0.75) {};
                    \node[fit=(a1) (b1), pair] {};
                \end{scope}

                \begin{scope}[yscale=.8, xshift=70mm, yshift=-5mm]
                    \draw[draw=none, rounded corners=0,fill=gray!20] (0,1.5)--(1,1.5)--(1,2.5)--(2,2.5)--(2,0.5)--(2.5,0.5)--(2.5,-0.5)--(3.5,-0.5)--
                    (3.5,-2.5)--(0,-2.5)--(0,1.5);

                    \foreach \n/\x/\y in {1/0/3, 2/1/3, 3/2/3, 4/2.5/1, 5/3.5/1}
                    \node[input] (x\n) at (\x, \y) {$x_{\n}$};
                    \node[gate] (g1) at (0.5,2) {$\oplus$};
                    \node[gate] (g2) at (1.5,2) {$\oplus$};
                    \node[gate] (g3) at (0.5,1) {$\lor$};
                    \node[gate] (g4) at (1.5,1) {$\oplus$};
                    \node[gate, label=left:$a_1$] (g5) at (0.5,0) {$\oplus$};
                    \node[gate] (g6) at (2,0) {$\oplus$};
                    \node[gate] (g7) at (3,0) {$\oplus$};
                    \node[gate] (g8) at (2,-1) {$>$};
                    \node[gate, label=right:$b_0$] (g9) at (3,-1) {$\oplus$};
                    \node[gate, label=right:$a_1 \oplus b_1$] (g10) at (2,-2) {$\oplus$};
                    \node[gate] (g11) at (1.5,-3) {$>$};

                    \foreach \f/\t in {x1/g1, x2/g1, x2/g2, x3/g2, g1/g3, g2/g3, g1/g4, g3/g5, g4/g5, x4/g6, g4/g6, x4/g7, x5/g7, g6/g8, g7/g8, g7/g9, g3/g10, g8/g10, g10/g11, g5/g11}
                    \draw[->] (\f) -- (\t);

                    \path (x3) edge[->,bend left] (g4);
                    \path (g4) edge[->,bend left=20] (g9);
                \end{scope}
            \end{tikzpicture}
        \end{center}
        \caption{Two consecutive Full Adders (top left), the MDFA block (bottom left), an~optimal circuit for $\SUM_5$ (top right) whose highlighted part computes~MDFA, and a~dot diagram for MDFA.}
        \label{figure:mdfa}
    \end{figure}

    We~also need a~block called MDFA' that can be~viewed as~a~subfunction of~MDFA:
    \[\operatorname{MDFA'}(x_1 \oplus x_2, x_2, x_4, x_4 \oplus x_5)=\operatorname{MDFA}(x_1 \oplus x_2, x_2, 1, x_4, x_4 \oplus x_5).\]
    It~is not difficult to~see that one can compute MDFA' using six gates: when one replaces $x_3$ by~one in~the circuit for MDFA,
    the two gates fed by~$x_3$ can be~eliminated.

    Using MDFA and MDFA' blocks, one can compute $\SUM_n$ roughly as~follows:
    \begin{enumerate}
        \item Compute $x_2 \oplus x_3, x_4 \oplus x_5, \dotsc, x_{n-1} \oplus x_n$ ($n/2$ gates).
        \item Apply at~most~$n/2$ $\MDFA$ blocks (no~more than $4n$~gates).
        \item The last MDFA block outputs two bits: $a$~and~$a\oplus b$. Instead of~them, one needs to~compute $a \oplus b$ and $a \land b$. To~achieve this,
        it~suffices to apply $x>y=(x \land \overline{y})$ operation:
        \(a \land b = a>(a \oplus b)\).
    \end{enumerate}
    This gives an~upper bound $4.5n$ for $\SUM_n$, its formal proof can
    be~found in~\cite{DBLP:journals/ipl/DemenkovKKY10}.
    Figure~\ref{figure:sum17mdfa} gives an~example of~the corresponding design
    for~$\SUM_{16}$.

    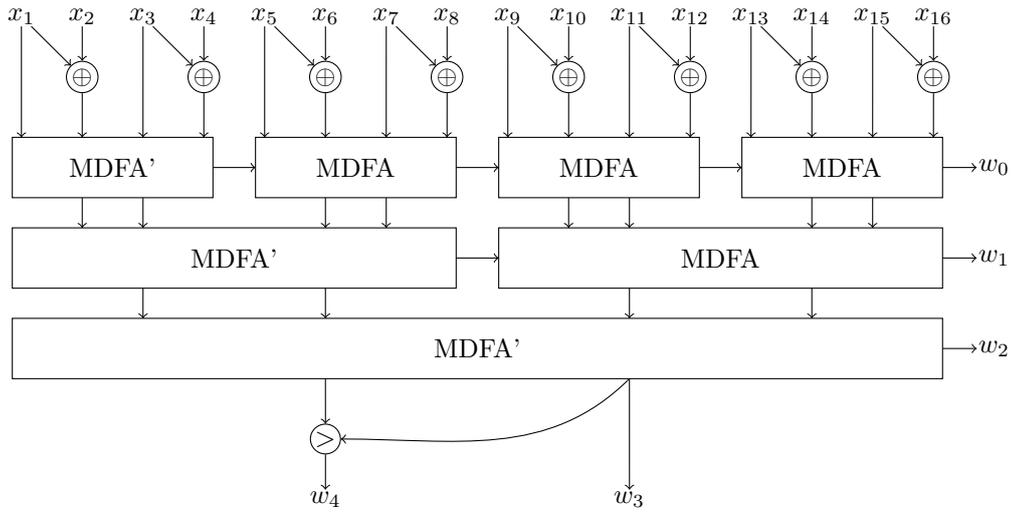
\begin{figure}[ht]
        \begin{center}
            \begin{tikzpicture}[scale=0.8]
                \foreach \n in {2,...,17} {
                    \tikzmath{\j=int(\n-1);}
                    \node[input] (\n) at (\n,6) {$x_{\j}$};
                }
                \foreach \i in {1,...,8} {
                    \tikzmath{\k=int(2*\i); \j=int(2*\i+1);}
                    \node[gate] (a) at (\j,5) {$\oplus$};
                    \draw[->] (\k) -- (a); \draw[->] (\j) -- (a);
                    \draw[->] (a) -- (\j,4); \draw[->] (\k) -- (\k,4);
                }
                \foreach \x/\y/\w/\t in {2/4/3/MDFA', 6/4/3/MDFA, 10/4/3/MDFA, 14/4/3/MDFA, 2/2.5/7/MDFA', 10/2.5/7/MDFA, 2/1/15/MDFA'} {
                    \draw (\x-0.15,\y) rectangle (\x+\w+0.15,\y-1);
                    \node at (\x+\w/2,\y-0.5) {\t};
                }

                \node[input] (c) at (18,3.5) {$w_0$};
                \draw[->] (17.15,3.5) -- (c);

                \node[input] (w_1) at (18, 2) {$w_1$}; \draw[->] (17.15, 2) -- (w_1);
                \node[input] (w_2) at (18, 0.5) {$w_2$}; \draw[->] (17.15, 0.5) -- (w_2);

                \foreach \x in {3, 4, 7, 8, 11, 12, 15, 16}
                \draw[->] (\x,3) -- (\x,2.5);
                \foreach \x in {4, 7, 12, 15}
                \draw[->] (\x,1.5) -- (\x,1);
                \foreach \x/\y in {5/3.5, 9/3.5, 13/3.5, 9/2}
                \draw[->] (\x+0.15,\y) -- (\x+0.85,\y);

                \node[input] (w3) at (12,-2) {$w_3$};
                \draw[->] (12,0) -- (w3);
                \node[gate] (x) at (7,-1) {$>$};
                \draw[->] (7,0) -- (x);
                \node[input] (w4) at (7,-2) {$w_4$};
                \draw[->] (x) -- (w4);
                \path (12,0) edge[->,out=-135,in=0] (x);
            \end{tikzpicture}
        \end{center}
        \caption{A~circuit computing $\SUM_{16}$ composed out~of eight $\oplus$-gates at~the top, three MDFA' blocks, four MDFA blocks, and one final gate. Its size is $8+3 \cdot 6 + 4 \cdot 8 + 1=59$.}
        \label{figure:sum17mdfa}
    \end{figure}

    \section{New Upper Bound for Circuit Size of~Bit Addition}
    In~this section, we~prove a~new upper bound $4.5n-2m$ for the circuit size
    of~bit addition. For regimes where $m$~is small compared to~$n$, this~is
    better than $5n-3m$ by~about $10\%$. This applies to~$\MULT_n$ and $\SUM_n$.
    \begin{theorem}
        \label{theorem:main}
        For any vector $s \in \mathbb{Z}_{\ge 0}^n$,
        \[\operatorname{size}(\BA_n^s) \le 4.5n-2m.\]
    \end{theorem}

    In~the proof, we~use the following straightforward observation.
    Assume that $s_1<s_2,s_3, \dotsc, s_n$.
    In~this case, the first output is~equal to~$x_1$,
    the cost of~computing this particular bit of~the output is~zero,
    allowing one to~forget about~it. Thus,
    \[\operatorname{size}(\BA^s_n)=\operatorname{size}(\BA^{s'}_{n-1}),\]
    where $s'=(s_2,\dotsc, s_n)$. We~call the operation
    of~replacing~$s$ by~$s'$
    as~\emph{shifting}. Note that shifting reduces both the number of~inputs and the number of~outputs by~one. Figure~\ref{figure:shifting} gives an~example.

    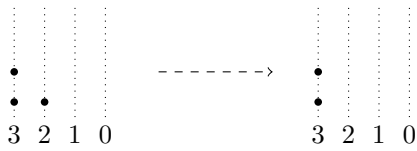
\begin{figure}[ht]
        \begin{center}
            \begin{tikzpicture}
                \begin{scope}
                    \foreach \x [count=\n from 0] in {2, 1, 0, -1} {
                        \draw[l] (\x * \d, -1) -- (\x * \d, .5);
                        \node[below] at (\x * \d, -1) {$\n$};
                    }

                    \foreach \x/\y in {-1/-1, -1/0, 0/-1}
                    	\node[dot] at (\x * \d, \y * \d - \d) {};
                \end{scope}

                \draw[dashed, ->] (1.5, -\d) -- (3, -\d);

                \begin{scope}[xshift=40mm]
                    \foreach \x [count=\n from 0] in {2, 1, 0, -1} {
						\draw[l] (\x * \d, -1) -- (\x * \d, .5);
						\node[below] at (\x * \d, -1) {$\n$};
					}

					\foreach \x/\y in {-1/-1, -1/0}
						\node[dot] at (\x * \d, \y * \d - \d) {};
				\end{scope}
            \end{tikzpicture}
        \end{center}
        \caption{Shifting: $\size(\BA_3^{2, 3, 3})=\size(\BA_2^{3, 3})$. In~turn, $\BA_2^{3, 3}$ can be~computed by the Half Adder. Thus, $\BA_3^{2, 3, 3}(x_1, x_2, x_3)=(x_1, x_2 \oplus x_3, x_2 \land x_3)$.}
        \label{figure:shifting}
    \end{figure}

    \begin{proof}
        As~the first step, we~do the following: at~every significance layer,
        we~break all bits, except for possibly one, into pairs and compute
        the parity for every pair. This takes at~most $n/2$ gates.

        \begin{center}
            \begin{tikzpicture}
                \foreach \x [count=\n from 0] in {2, 1, ..., -2} {
                    \draw[l] (\x * \d, -1) -- (\x * \d, 1);
                }

                \foreach [count=\n] \x/\y in {-2/-1, -2/0, -2/1, -2/2, -2/3, -1/-1, -1/0, -1/1, 0/-1, 1/-1, 1/0, 1/1, 1/2, 2/-1, 2/0}
                	\node[dot] (\n) at (\x * \d, \y * \d - \d) {};

                \foreach \i/\j in {1/2, 3/4, 6/7, 10/11, 12/13, 14/15}
                	\node[fit=(\i) (\j), pair] {};
            \end{tikzpicture}
        \end{center}

        Then, it~remains to~prove that one can compute the sum of~$n$ bits
        using $4n-2m$ gates if~every significance layer contains at~most one bit
        without a~pair. We~prove this by~induction on~$n$. The base case $n=1$ is~clear: in~this case, the circuit size is~zero (nothing needs to~be summed~up) and the upper bound is~at~least zero since $m \le n$.
        To~prove the induction step,
        denote by~$l$ the number
        of~minimum elements in~the significance vector~$s$
        (that~is, the number of~bits in~the rightmost non-empty column in~dot-notation).

        Consider the following seven cases. (In~fact, the first three cases are special cases of~the last three cases, but we~believe that the presentation is~cleaner when they are stated
        as~separate cases.)
        In~each of~the cases below, we~shift and proceed by~induction.

        \begin{enumerate}
            \item $l=1$. In~this case, we~just shift.
            By~the induction hypothesis, the rest can be~computed by~a~circuit
            of~size at~most
            \[4(n-1)-2(m-1)=4n-2m-2<4n-2m.\]

            \item $l=2$. Then, the corresponding two bits~$x_1$ and~~$x_2$ are paired meaning that their sum $x_1 \oplus x_2$ is~computed already.
            Then, we~compute their carry
            \[c=x_1 > (x_1 \oplus x_2)=x_1 \land x_2\]
            and transfer~it to~the next layer.
            If~this layer has an~unpaired bit~$b$, we~pair $b$~and~$c$
            by~computing $b \oplus c$. Finally, we~shift.
            By~the induction hypothesis, the size of~the resulting circuit is~at~most
            \[1+1+4(n-1)-2(m-1)=4n-2m.\]

            \begin{center}
                \begin{tikzpicture}
                    \begin{scope}
                        \foreach \x [count=\n from 0] in {2, 1} {
                            \draw[l] (\x * \d, -1) -- (\x * \d, 0.5);
                        }

                        \foreach [count=\n] \x/\y in {2/-1, 2/0, 1/-1, 1/0, 1/1}
                        \node[dot] (\n) at (\x * \d, \y * \d - \d) {};

                        \foreach \i/\j in {1/2, 3/4}
                        \node[fit=(\i) (\j), pair] {};
                    \end{scope}

                    \begin{scope}[xshift=20mm]
                        \foreach \x [count=\n from 0] in {2, 1} {
                            \draw[l] (\x * \d, -1) -- (\x * \d, 0.5);
                        }

                        \foreach [count=\n] \x/\y in {2/-1, 1/-1, 1/0, 1/1, 1/2}
                        \node[dot] (\n) at (\x * \d, \y * \d - \d) {};

                        \foreach \i/\j in {2/3}
                        \node[fit=(\i) (\j), pair] {};
                    \end{scope}

                    \begin{scope}[xshift=40mm]
                        \foreach \x [count=\n from 0] in {2, 1} {
                            \draw[l] (\x * \d, -1) -- (\x * \d, 0.5);
                        }

                        \foreach [count=\n] \x/\y in {2/-1, 1/-1, 1/0, 1/1, 1/2}
                        \node[dot] (\n) at (\x * \d, \y * \d - \d) {};

                        \foreach \i/\j in {2/3, 4/5}
                        \node[fit=(\i) (\j), pair] {};
                    \end{scope}

                    \begin{scope}[xshift=60mm]
                        \foreach \x [count=\n from 0] in {2} {
                            \draw[l] (\x * \d, -1) -- (\x * \d, 0.5);
                        }

                        \foreach [count=\n] \x/\y in {2/-1, 2/0, 2/1, 2/2}
                        \node[dot] (\n) at (\x * \d, \y * \d - \d) {};

                        \foreach \i/\j in {1/2, 3/4}
                        \node[fit=(\i) (\j), pair] {};
                    \end{scope}

                    \path (1.2, -0.5) edge[->, dashed] node[below] {$\land$} (2, -0.5);
                    \path (3.2, -0.5) edge[->, dashed] node[below] {$\oplus$} (4, -0.5);
                    \path (5.2, -0.5) edge[->, dashed] node[below] {shift} (6, -0.5);
                \end{tikzpicture}
            \end{center}

            \item $l=3$. For the corresponding three bits $x_1,x_2,x_3$,
            we~have $x_1 \oplus x_2$, $x_2$, and $x_3$ (that~is, $x_1$~and~$x_2$ are paired). We~apply the Full Adder to~the three bits. This costs four gates, as~$x_1 \oplus x_2$ is~already computed and $x_1$ is~not needed
            (recall Figure~\ref{figure:sum23}). The sum bit stays on~the same layer, whereas the carry bit~$c$ goes to~the next layer. Then, we~pair~$c$
            with an~unpaired bit on~the next layer if~needed and shift. This gives an~upper bound
            \[4+1+4(n-2)-2(m-1)<4n-2m.\]

            \begin{center}
                \begin{tikzpicture}
                    \begin{scope}
                        \foreach \x [count=\n from 0] in {2, 1} {
                            \draw[l] (\x * \d, -1) -- (\x * \d, 0.5);
                        }

                        \foreach [count=\n] \x/\y in {2/-1, 2/0, 2/1, 1/-1, 1/0, 1/1}
                        \node[dot] (\n) at (\x * \d, \y * \d - \d) {};

                        \foreach \i/\j in {1/2, 4/5}
                        \node[fit=(\i) (\j), pair] {};
                    \end{scope}

                    \begin{scope}[xshift=20mm]
                        \foreach \x [count=\n from 0] in {2, 1} {
                            \draw[l] (\x * \d, -1) -- (\x * \d, 0.5);
                        }

                        \foreach [count=\n] \x/\y in {2/-1, 1/-1, 1/0, 1/1, 1/2}
                        \node[dot] (\n) at (\x * \d, \y * \d - \d) {};

                        \foreach \i/\j in {2/3}
                        \node[fit=(\i) (\j), pair] {};
                    \end{scope}

                    \begin{scope}[xshift=40mm]
                        \foreach \x [count=\n from 0] in {2, 1} {
                            \draw[l] (\x * \d, -1) -- (\x * \d, 0.5);
                        }

                        \foreach [count=\n] \x/\y in {2/-1, 1/-1, 1/0, 1/1, 1/2}
                        \node[dot] (\n) at (\x * \d, \y * \d - \d) {};

                        \foreach \i/\j in {2/3, 4/5}
                        \node[fit=(\i) (\j), pair] {};
                    \end{scope}

                    \begin{scope}[xshift=60mm]
                        \foreach \x [count=\n from 0] in {2} {
                            \draw[l] (\x * \d, -1) -- (\x * \d, 0.5);
                        }

                        \foreach [count=\n] \x/\y in {2/-1, 2/0, 2/1, 2/2}
                        \node[dot] (\n) at (\x * \d, \y * \d - \d) {};

                        \foreach \i/\j in {1/2, 3/4}
                        \node[fit=(\i) (\j), pair] {};
                    \end{scope}

                    \path (1.2, -0.5) edge[->, dashed] node[below] {FA} (2, -0.5);
                    \path (3.2, -0.5) edge[->, dashed] node[below] {$\oplus$} (4, -0.5);
                    \path (5.2, -0.5) edge[->, dashed] node[below] {shift} (6, -0.5);
                \end{tikzpicture}
            \end{center}

            \item $l=4k$. Apply MDFA' to~two pairs to~produce an~unpaired bit.
            For the remaining $2k-2$ pairs, keep applying MDFA, each time reusing the unpaired bit. Then, we~shift. The upper bound~is
            \[6+8(k-1)+4(n-2k)-2(m-1)=4n-2m.\]

            \begin{center}
                \begin{tikzpicture}
                    \begin{scope}
                        \foreach \x [count=\n from 0] in {2, 1} {
                            \draw[l] (\x * \d, -1) -- (\x * \d, 2.5);
                        }

                        \foreach [count=\n] \x/\y in {2/-1, 2/0, 2/1, 2/2, 2/3, 2/4, 2/5, 2/6, 1/-1, 1/0}
                        \node[dot] (\n) at (\x * \d, \y * \d - \d) {};

                        \foreach \i/\j in {1/2, 3/4, 5/6, 7/8, 9/10}
                        \node[fit=(\i) (\j), pair] {};
                    \end{scope}

                    \begin{scope}[xshift=20mm]
                        \foreach \x [count=\n from 0] in {2, 1} {
                            \draw[l] (\x * \d, -1) -- (\x * \d, 2.5);
                        }

                        \foreach [count=\n] \x/\y in {2/-1, 2/0, 2/1, 2/2, 1/-1, 1/0, 1/1, 1/2, 2/3}
                        \node[dot] (\n) at (\x * \d, \y * \d - \d) {};

                        \foreach \i/\j in {1/2, 3/4, 5/6, 7/8}
                        \node[fit=(\i) (\j), pair] {};
                    \end{scope}

                    \begin{scope}[xshift=40mm]
                        \foreach \x [count=\n from 0] in {2, 1} {
                            \draw[l] (\x * \d, -1) -- (\x * \d, 2.5);
                        }

                        \foreach [count=\n] \x/\y in {1/-1, 1/0, 1/1, 1/2, 1/3, 1/4, 2/-1}
                        \node[dot] (\n) at (\x * \d, \y * \d - \d) {};

                        \foreach \i/\j in {1/2, 3/4, 5/6}
                        \node[fit=(\i) (\j), pair] {};
                    \end{scope}

                    \begin{scope}[xshift=60mm]
                        \foreach \x [count=\n from 0] in {2} {
                            \draw[l] (\x * \d, -1) -- (\x * \d, 2.5);
                        }

                        \foreach [count=\n] \x/\y in {2/-1, 2/0, 2/1, 2/2, 2/3, 2/4}
                        \node[dot] (\n) at (\x * \d, \y * \d - \d) {};

                        \foreach \i/\j in {1/2, 3/4, 5/6}
                        \node[fit=(\i) (\j), pair] {};
                    \end{scope}

                    \path (1.2, -0.5) edge[->, dashed] node[below] {MDFA'} (2, -0.5);
                    \path (3.2, -0.5) edge[->, dashed] node[below, text width=20mm, align=center] {MDFA\\ $k-1$} (4, -0.5);
                    \path (5.2, -0.5) edge[->, dashed] node[below] {shift} (6, -0.5);
                \end{tikzpicture}
            \end{center}

            \item $l=4k+1$. Apply MDFA $k$~times, then shift. The upper bound~is
            \[8k+4(n-2k-1)-2(m-1)<4n-2m.\]

            \begin{center}
                \begin{tikzpicture}
                    \begin{scope}
                        \foreach \x [count=\n from 0] in {2, 1} {
                            \draw[l] (\x * \d, -1) -- (\x * \d, 2.5);
                        }

                        \foreach [count=\n] \x/\y in {2/-1, 2/0, 2/1, 2/2, 2/3, 2/4, 2/5, 2/6, 2/7, 1/-1, 1/0}
                        \node[dot] (\n) at (\x * \d, \y * \d - \d) {};

                        \foreach \i/\j in {1/2, 3/4, 5/6, 7/8, 10/11}
                        \node[fit=(\i) (\j), pair] {};
                    \end{scope}

                    \begin{scope}[xshift=20mm]
                        \foreach \x [count=\n from 0] in {2, 1} {
                            \draw[l] (\x * \d, -1) -- (\x * \d, 2.5);
                        }

                        \foreach [count=\n] \x/\y in {1/-1, 1/0, 1/1, 1/2, 1/3, 1/4, 2/-1}
                        \node[dot] (\n) at (\x * \d, \y * \d - \d) {};

                        \foreach \i/\j in {1/2, 3/4, 5/6}
                        \node[fit=(\i) (\j), pair] {};
                    \end{scope}

                    \begin{scope}[xshift=40mm]
                        \foreach \x [count=\n from 0] in {2} {
                            \draw[l] (\x * \d, -1) -- (\x * \d, 2.5);
                        }

                        \foreach [count=\n] \x/\y in {2/-1, 2/0, 2/1, 2/2, 2/3, 2/4}
                        \node[dot] (\n) at (\x * \d, \y * \d - \d) {};

                        \foreach \i/\j in {1/2, 3/4, 5/6}
                        \node[fit=(\i) (\j), pair] {};
                    \end{scope}

                    \path (1.2, -0.5) edge[->, dashed] node[below, text width=20mm, align=center] {MDFA\\$k$} (2, -0.5);
                    \path (3.2, -0.5) edge[->, dashed] node[below] {shift} (4, -0.5);
                \end{tikzpicture}
            \end{center}

            \item $l=4k+2$. Compute an~$\land$ of~two bits from the same pair:
            this leaves their sum at~the current layer and puts the just computed
            carry bit to~the next layer. If~needed, compute the parity of~an~unpaired
            bit with the just transferred carry bit. Then, apply MDFA $k$~times and shift.
            Overall, the upper bound~is
            \[1+1+8k+4(n-2k-1)-2(m-1)=4n-2m.\]

            \begin{center}
                \begin{tikzpicture}
                    \begin{scope}
                        \foreach \x [count=\n from 0] in {2, 1} {
                            \draw[l] (\x * \d, -1) -- (\x * \d, 1.5);
                        }

                        \foreach [count=\n] \x/\y in {2/-1, 2/0, 2/1, 2/2, 2/3, 2/4, 1/-1, 1/0, 1/1}
                        \node[dot] (\n) at (\x * \d, \y * \d - \d) {};

                        \foreach \i/\j in {1/2, 3/4, 5/6, 7/8}
                        \node[fit=(\i) (\j), pair] {};
                    \end{scope}

                    \begin{scope}[xshift=20mm]
                        \foreach \x [count=\n from 0] in {2, 1} {
                            \draw[l] (\x * \d, -1) -- (\x * \d, 1.5);
                        }

                        \foreach [count=\n] \x/\y in {2/-1, 2/0, 2/1, 2/2, 1/-1, 1/0, 1/1, 1/2, 2/3}
                        \node[dot] (\n) at (\x * \d, \y * \d - \d) {};

                        \foreach \i/\j in {1/2, 3/4, 5/6, 7/8}
                        \node[fit=(\i) (\j), pair] {};
                    \end{scope}

                    \begin{scope}[xshift=40mm]
                        \foreach \x [count=\n from 0] in {2, 1} {
                            \draw[l] (\x * \d, -1) -- (\x * \d, 1.5);
                        }

                        \foreach [count=\n] \x/\y in {1/-1, 1/0, 1/1, 1/2, 1/3, 1/4, 2/-1}
                        \node[dot] (\n) at (\x * \d, \y * \d - \d) {};

                        \foreach \i/\j in {1/2, 3/4, 5/6}
                        \node[fit=(\i) (\j), pair] {};
                    \end{scope}

                    \begin{scope}[xshift=60mm]
                        \foreach \x [count=\n from 0] in {2} {
                            \draw[l] (\x * \d, -1) -- (\x * \d, 1.5);
                        }

                        \foreach [count=\n] \x/\y in {2/-1, 2/0, 2/1, 2/2, 2/3, 2/4}
                        \node[dot] (\n) at (\x * \d, \y * \d - \d) {};

                        \foreach \i/\j in {1/2, 3/4, 5/6}
                        \node[fit=(\i) (\j), pair] {};
                    \end{scope}

                    \path (1.2, -0.5) edge[->, dashed] node[below] {$\land$, $\oplus$} (2, -0.5);
                    \path (3.2, -0.5) edge[->, dashed] node[below, text width=20mm, align=center] {MDFA\\ $k$} (4, -0.5);
                    \path (5.2, -0.5) edge[->, dashed] node[below] {shift} (6, -0.5);
                \end{tikzpicture}
            \end{center}

            \item $l=4k+3$. Apply the Full Adder to~a~pair of~bits and the unpaired bit.
            If~needed, pair the just transferred carry bit with an~unpaired bit from
            the next layer. Then, apply MDFA $k$~times and shift. The resulting upper
            bound~is
            \[4+1+8k+4(n-2k-2)-2(m-1)<4n-2m.\]

            \begin{center}
                \begin{tikzpicture}
                    \begin{scope}
                        \foreach \x [count=\n from 0] in {2, 1} {
                            \draw[l] (\x * \d, -1) -- (\x * \d, 2);
                        }

                        \foreach [count=\n] \x/\y in {2/-1, 2/0, 2/1, 2/2, 2/3, 2/4, 1/-1, 1/0, 1/1, 2/5}
                        \node[dot] (\n) at (\x * \d, \y * \d - \d) {};

                        \foreach \i/\j in {1/2, 3/4, 5/6, 7/8}
                        \node[fit=(\i) (\j), pair] {};
                    \end{scope}

                    \begin{scope}[xshift=20mm]
                        \foreach \x [count=\n from 0] in {2, 1} {
                            \draw[l] (\x * \d, -1) -- (\x * \d, 2);
                        }

                        \foreach [count=\n] \x/\y in {2/-1, 2/0, 2/1, 2/2, 1/-1, 1/0, 1/1, 1/2, 2/3}
                        \node[dot] (\n) at (\x * \d, \y * \d - \d) {};

                        \foreach \i/\j in {1/2, 3/4, 5/6, 7/8}
                        \node[fit=(\i) (\j), pair] {};
                    \end{scope}

                    \begin{scope}[xshift=40mm]
                        \foreach \x [count=\n from 0] in {2, 1} {
                            \draw[l] (\x * \d, -1) -- (\x * \d, 2);
                        }

                        \foreach [count=\n] \x/\y in {1/-1, 1/0, 1/1, 1/2, 1/3, 1/4, 2/-1}
                        \node[dot] (\n) at (\x * \d, \y * \d - \d) {};

                        \foreach \i/\j in {1/2, 3/4, 5/6}
                        \node[fit=(\i) (\j), pair] {};
                    \end{scope}

                    \begin{scope}[xshift=60mm]
                        \foreach \x [count=\n from 0] in {2} {
                            \draw[l] (\x * \d, -1) -- (\x * \d, 2);
                        }

                        \foreach [count=\n] \x/\y in {2/-1, 2/0, 2/1, 2/2, 2/3, 2/4}
                        \node[dot] (\n) at (\x * \d, \y * \d - \d) {};

                        \foreach \i/\j in {1/2, 3/4, 5/6}
                        \node[fit=(\i) (\j), pair] {};
                    \end{scope}

                    \path (1.2, -0.5) edge[->, dashed] node[below] {FA, $\oplus$} (2, -0.5);
                    \path (3.2, -0.5) edge[->, dashed] node[below, text width=20mm, align=center] {MDFA\\ $k$} (4, -0.5);
                    \path (5.2, -0.5) edge[->, dashed] node[below] {shift} (6, -0.5);
                \end{tikzpicture}
            \end{center}
        \end{enumerate}
    \end{proof}

    \section{Lower Bounds and Limitations}
    The upper bound $\size(\BA_n^s) \le 5n-3m$ holds for any vector~$s$, but
    in~many scenarios it~is loose. For example, for the function
    \[\ADD_n=\BA_{2n}^{0,\dotsc,n-1,0,\dotsc,n-1},\]
    this upper bounds turns into $5\cdot 2n-3(n+1)=7n-3$,
    whereas $\size(\ADD_n)=5n-3$. Interestingly, the term $3m$ in~the upper bound $5n-3m$ cannot be~increased: for any constant $\alpha > 3$, there exists
    a~vector~$s$ such that $\size(\BA_n^s) \ge 5n-\alpha m-O(1)$.
    One example of~such a~vector is~$s=(0,0,1,\dotsc,n-1)$. The corresponding function $\BA_n^s$ adds a~bit to an~$n$-bit number. It~is not difficult to~see that it~can be~computed using $n$~Half Adders (see Figure~\ref{figure:addingbit}), hence its circuit size is~at~most~$2n$. Below, we~show that this straightforward circuit is~optimal (up~to an~additive constant). It~also shows that it~is
    impossible to~improve our upper bound $4.5n-2m$ to $4.5n-\beta m$ for $\beta > 2.5$.

    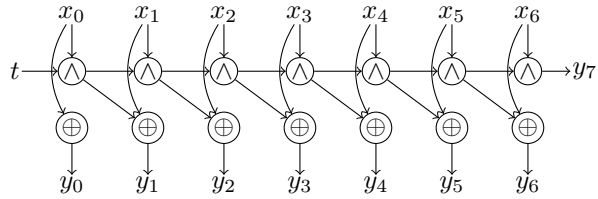
\begin{figure}
        \begin{center}
            \begin{tikzpicture}
                \foreach \i in {0, ..., 6} {
                    \node[input] (x\i) at (\i, 3) {$x_{\i}$};
                    \node[gate] (and\i) at (\i, 2.25) {$\land$};
                    \node[gate] (xor\i) at (\i, 1.5) {$\oplus$};
                    \node[input] (y\i) at (\i, 0.74) {$y_{\i}$};
                    \draw[->] (xor\i) -- (y\i);
                    \draw[->] (x\i) -- (and\i);
                    \path (x\i) edge[->, draw, bend right] (xor\i);
                }
                \foreach \i in {0,...,5} {
                    \tikzmath{\j=int(\i+1);}
                    \draw[->] (and\i) -- (and\j);
                    \draw[->] (and\i) -- (xor\j);
                }

                \node[input] (t) at (-0.75, 2.25) {$t$};
                \node[input] (y7) at (6.75, 2.25) {$y_7$};
                \draw[->] (t) -- (and0);
                \draw[->] (and6) -- (y7);
            \end{tikzpicture}
        \end{center}
        \caption{One can add a~bit~$t$ to a~$7$-bit integer $x_0\dotsb x_7$ (to~get
            an~$8$-bit integer $y_0\dotsb y_7$) using $14$~gates. A~straightforward
            generalization of~this construction ensures that $\size(\BA_n^{0,0,1,\dotsc, n-1}) \le 2n$.}
        \label{figure:addingbit}
    \end{figure}

    \begin{theorem}
        \label{theorem:lowerbound}
        \[\size(\BA_n^{0,0,1,\dotsc,n-1}) \ge 2n-O(1).\]
    \end{theorem}
    \begin{proof}
        Assume that a~bit to~be added is~equal to~one (clearly, this only makes the function easier to~compute). In~other words,
        we~consider the \emph{increment} function $\INC_n \colon \{0,1\}^n \to \{0,1\}^{n+1}$. Thus, $\INC_n(x_0, \dotsc, x_{n-1})=(y_0, \dotsc, y_n)$ such that
        \[1+\sum_{i=0}^{n-1}2^ix_i=\sum_{i=0}^n2^iy_i.\]
        It~is not difficult to~write down explicit formulas for all output bits of~$\INC_n$. For example, for $n=5$, they are expressed as~follows:
        \begin{align*}
            y_0 &= 1 \oplus x_0\\
            y_1 &= x_0 \oplus x_1\\
            y_2 &= x_0x_1 \oplus x_2\\
            y_3 &= x_0x_1x_2 \oplus x_3\\
            y_4 &= x_0x_1x_2x_3 \oplus x_4\\
            y_5 &= x_0x_1x_2x_3x_4\\
        \end{align*}

        We~prove that $\size(\INC_n) \ge 2n-2$ by~induction on~$n$.
        The base case $n=1$ is~clear. For the induction step, take an~optimal circuit computing $\INC_n$ and consider its (topologically) first gate~$A(x_i, x_j)$.

        Now, if~both the variables $x_i$~and~$x_j$ had out-degree one, the whole circuit would depend on~$x_i$ and~$x_j$ through the gate~$A$ only. This would mean that there are two different pairs of~constant $(a_i, a_j), (b_i, b_j) \in \{0,1\}^2$ such that $A(a_i, a_j)=A(b_i, b_j)$. This, in~turn, would mean that the circuit does not distinguish between assignments
        \[x_i \gets a_i,\, x_j \gets a_j \,\text{ and }\, x_i \gets b_i,\, x_j \gets b_j.\]
        But such a~circuit cannot compute the function $\INC_n$ as~$\INC_n$ clearly distinguishes all four different assignments to~$x_i$ and~$x_j$.

        Thus, assume that, say, $x_i$ feeds at~least two gates. Then, assign $x_i \gets 1$ and simplify the circuit. During the simplification, the gates fed by~$x_i$ are eliminated. Also, the resulting circuit computes $\INC_{n-1}$. To~see this,
        it~is instructive
        to~get back to~the previous toy example where $n=5$. Say, we assign $x_2 \gets 1$. Then, the outputs are simplified as~follows:
        \begin{align*}
            y_0 &= 1 \oplus x_0\\
            y_1 &= x_0 \oplus x_1\\
            y_2 &= x_0x_1 \oplus 1\\
            y_3 &= x_0x_1 \oplus x_3\\
            y_4 &= x_0x_1x_3 \oplus x_4\\
            y_5 &= x_0x_1x_3x_4\\
        \end{align*}
        By~ignoring the output~$y_2$, one gets a~function computing $\INC_4$:
        \[(y_0,y_1,y_3,y_4,y_5)=\INC_4(x_0, x_1, x_3, x_4)\,.\]

        By~the induction hypothesis, the simplified circuit contains at~least $2(n-1)-2=2n-4$ gates. Thus, the original circuit has size at~least $2+2n-4=2n-2$.
    \end{proof}

    \section{Implementation and Experimental Evaluation}

    We~implemented efficient generators of~our new circuits in~the \texttt{Cirbo} open-source framework~\cite{DBLP:conf/aaai/AverkovBEGKKKLL25}.
    To~generate a~circuit computing $\BA_n^s$, one passes the vector~$s$. Listing~\ref{listing:first} shows how to~use the generator
    to~produce an~efficient circuit computing $\SUM_{31}$ in~a~single line of~code.
    When the circuit is~generated, one can use a~wide range of~\texttt{Cirbo} methods to~analyze and manipulate the circuit.

    \begin{lstlisting}[
        caption={Generating an~efficient circuit for $\SUM_{31}$ (that computes the binary representation of~the sum of~$31$ bits). The code also prints the size
        of~the resulting circuit and draws~it.},
        label=listing:first,
        float=h,
        %abovecaptionskip=-\medskipamount
    ]
from cirbo.synthesis.generation.arithmetics.summation
    import generate_add_weighted_bits_efficient

ckt = generate_add_weighted_bits_efficient([0] * 31)
print(ckt.gates_number())
ckt.view_graph()
    \end{lstlisting}

    \subsection{Adding Bits and Integers}
    Table~\ref{table:first} compares the size of~circuits for $\SUM_n$
    composed out~of Full Adders with circuits composed out~of MDFA blocks
    (that can be~generated using our new method), for various~$n$. As~the table reveals,
    for large values of~$n$, the latter circuits are about $10\%$ smaller than the former ones. Also, Listing~\ref{listing:add} ensures that for the addition of~two $n$-bit integers the generator produces circuits of~size $5n-3$ (recall that $\ADD_n=\BA_{2n}^{0,\dotsc,n-1,0,\dotsc,n-1}$ and that $5n-3$ is~provably optimal circuit size for this function).

    \begin{table}[ht]
        \caption{Comparing the size of~circuits for $\SUM_n$ composed out of~Full Adders with
            circuits composed out of~MDFA. The bottom row shows the improvement in~percents.}
        \label{table:first}
        \begin{center}
            \begin{tabular}{lrrrrrrrrr}
                \toprule
                $n$ & 7 & 31 & 127 & 511 & 2047 & 8191 & 32767 & 131071 \\
                \midrule
                FA & 20 & 130 & 600 & 2510 & 10180 & 40890 & 163760 & 655270 \\
                MDFA & 19 & 119 & 543 & 2263 & 9167 & 36807 & 147391 & 589751 \\
                Improvement  & 5.0\%  & 8.5\%  & 9.5\%  & 9.8\%  & 10.0\%  & 10.0\%  & 10.0\%  & 10.0\% \\
                \bottomrule
            \end{tabular}
        \end{center}
    \end{table}

        \begin{lstlisting}[
            caption={Ensuring that the generator produces circuits of~size $5n-3$ for $\ADD_n$.},
            label=listing:add,
            float=ht,
        ]
from cirbo.synthesis.generation.arithmetics.summation
    import generate_add_weighted_bits_efficient as generate

for n in range(2, 100):
    ckt = generate(list(range(n)) + list(range(n)))
    assert ckt.gates_number() == 5 * n - 3
        \end{lstlisting}

    \subsection{Multiplying Integers}

    Dadda's multiplier is~one of~the first circuit designs for multiplying $n$-bits integers. Basically, it~adds the partial products (conjunctions of~the bits of~the
    two input numbers) using Full Adders and Half Adders. Its size is~about $n^2+5n^2=6n^2$. Our method of~summing~up bits allows to~reduce the size
    to~roughly $5.5n^2$. An~asymptotically faster method of~multiplying
    $n$-bit integers was discovered by~Karatsuba~\cite{karatsuba}. It~is based
    on~the divide-and-conquer approach: to~multiply two $n$-bit integers,
    it~makes three recursive calls to~multiply two $n/2$-bit integers and
    then combines them using summation and subtraction only. This way, the running
    time~$T(n)$ of~the algorithm satisfies a~recurrence $T(n) \le 3T(n/2)+O(n)$,
    hence $T(n)=O(n^{\log_23})$. As~with many other algorithms based on~divide-and-conquer, when $n$~becomes small, it~is beneficial to~multiply
    the given numbers directly (rather than recursively).
    We~implemented generators based on~Karatsuba algorithm that
    use Dadda multipliers and MDFA-based bit addition when $n$~is smaller than~$20$.
    Table~\ref{table:multiplication} and Figure~\ref{figure:multiplication}
    compare the size of~the corresponding five circuit designs for $40 \le n \le 300$. More accurate estimates on~the circuits size of~multipliers
    are given by~Sergeev~\cite{Ser16}.

    \begin{table}[ht]
        \caption{Comparing the size of~circuits for $\MULT_n$. The first multiplier, Dadda, computes the sum of~the partial products using Full Adders and Half Adders. The second one, MDFA, sums~up the partial products using MDFA blocks.
        The third one, Karatsuba, makes three recursive calls and recurses all the way down to~4-bit numbers. The fourth and fifth multipliers use Karatsuba algorithm, but switch to~Dadda or~MDFA multipliers when $n$~becomes smaller than~$20$. The last row shows size improvement of~the fifth multiplier over the fourth one.}
        \label{table:multiplication}
        \begin{center}
            \begin{tabular}{lrrrrrrrrrrrrrrrr}
                \toprule
                $n$ & 40 & 80 & 120 & 160 & 200 & 240 & 280 \\
                \midrule
                Dadda & 9280 & 37760 & 85440 & 152320 & 238400 & 343680 & 468160 \\
                MDFA & 8539 & 34679 & 78419 & 139759 & 218699 & 315239 & 429379 \\
                Karatsuba & 11789 & 37836 & 72209 & 118152 & 168200 & 223093 & 293405 \\
                Karatsuba, Dadda& 7522 & 24770 & 49200 & 78598 & 113870 & 153948 & 199102 \\
                Karatsuba, MDFA& 7155 & 23642 & 46504 & 75168 & 108284 & 145787 & 189657 \\
                Improvement  & 4.9\%  & 4.6\%  & 5.5\%  & 4.4\%  & 4.9\%  & 5.3\%  & 4.7\% \\
                \bottomrule
            \end{tabular}
        \end{center}
    \end{table}

    \begin{figure}[ht]
        \includegraphics[width=\linewidth]{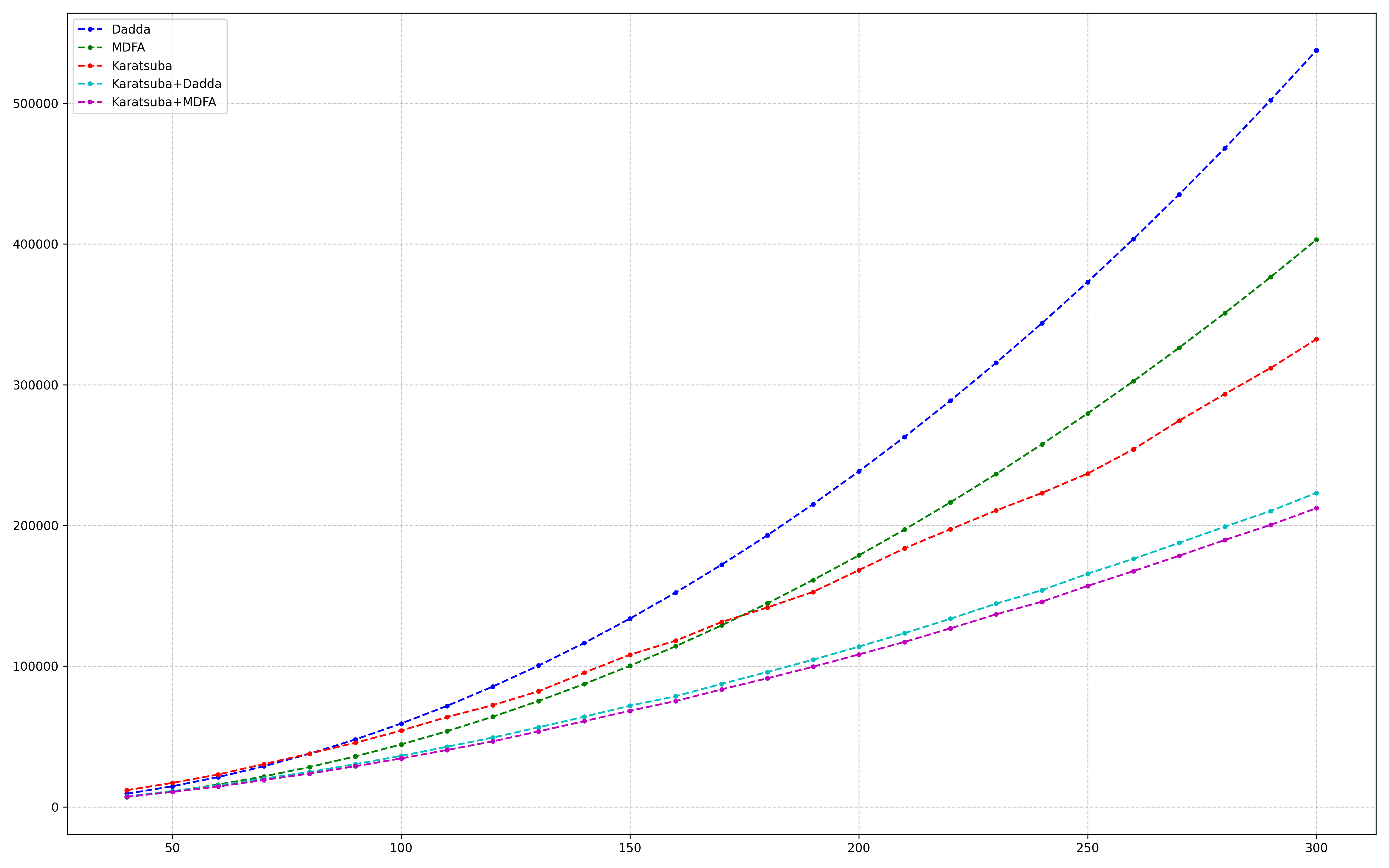}
        \caption{Comparing the size of~five circuit designs for $\MULT_n$, for $40 \le n \le 300$.}
        \label{figure:multiplication}
    \end{figure}

    \subsection{Logarithmic Depth}
    The depth of~most of~the circuits described above
    is~linear, that~is, $\Theta(n)$.
    With an~additional care, one can make the depth logarithmic ($\Theta(\log n)$) by~increasing the size slightly.
    To~achieve this, one processes the layers in~parallel rather than consecutively.
    Namely, let $h$~be the maximum height of~a~significance layer (that~is, every layer contains at~most $h$~bits). While $h > 3$, apply in~parallel as~many
    FA's as~possible to~every layer. After one such step, the maximum height becomes at~most $2h/3$ (to~simplify the exposition, we~ignore constant additive terms here): indeed, if~there are $k \le h$ bits on~the current layer, then about $k/3 \le h/3$ bits remain after the application of FA's; also, at~most $h/3$ bits are transferred from the next layer. Since the maximum height decreases geometrically, in~at most $O(\log n)$ steps, one reaches the case when $h \le 3$.
    This takes depth $O(\log n)$ and size $O(n)$ (since each FA reduces the total number of~bits by~one). When $h \le 3$, apply either HA~or~FA to~every layer. This ensures that every layer has at~most two bits, that~is, $h \le 2$ (the size of~the resulting circuit is~still linear and the depth is~still logarithmic).
    Then, everything boils down to~adding
    two $k$-bit numbers (with $k \le n$). This can
    be~performed using, for example,
    the Brent--Kung adder~\cite{BrentKung1982} that has size $O(k)$ and depth $O(\log k)$. By~using MDFA's instead of~FA's, one can further reduce the size of~the resulting circuits. Table~\ref{table:log} shows the size and the depth of the circuits generated this way for the three previously considered functions: $\SUM$, $\ADD$, and $\MULT$.

	\begin{table*}[ht]
		\caption{The size and the depth of~circuits computing $\SUM_n$, $\ADD_n$, and $\MULT_n$.}
		\label{table:log}
		\begin{center}
			\begin{tabular}{{lrrrrrrrrl}}
				\toprule
				$n$ & 10 & 20 & 30 & 40 & 60 & 80 & 160 & 320 \\
				\midrule
				\multirow{2}{*}{ADD}  & 15 & 18	& 23 & 24 & 28 & 31 & 32 & 42 & depth \\
				& 49 & 101 & 153 & 194 & 297 & 383 & 755 & 1526 & size \\
				\midrule
				\multirow{2}{*}{SUM}  & 10 & 14 & 16 & 18 & 20 & 22 & 26 & 30 & depth \\
				& 64 & 141 & 215 & 298 & 452 & 615 & 1252 & 2529 & size \\
				\midrule
				\multirow{2}{*}{MULT} & 29 & 39 & 46 & 51 & 55 & 62 & 70 & 81 & depth \\
				& 627 & 2301 & 5209 & 9158 & 20356 & 35971 & 142388 & 566539 & size \\
				\bottomrule
			\end{tabular}
		\end{center}
	\end{table*}

    \section{Conclusion and Further Directions}
    In this paper, we~presented smaller circuits for bit addition.
    In~many practically relevant scenarios, the described circuits
    are about 10\% smaller than the known circuits composed
    out~of Half Adders and Full Adders.
    Also, we~implemented generators that allow one
    to~produce the corresponding circuits using a~single line of~code
    via the \texttt{Cirbo} open-source package~\cite{DBLP:conf/aaai/AverkovBEGKKKLL25}.

    There are three natural open problems related to~the circuit size of~bit addition.
    \begin{enumerate}
        \item What is the largest $\alpha$ such that
        \[\size(\BA_n^s) \le 4.5n-\alpha m\]
        holds for every vector~$s$? In~this paper, we~proved that $\alpha \ge 2$.
        Theorem~\ref{theorem:main} shows that $\alpha \le 2.5$. An~example
        of~a~vector where our upper bound $4.5n-2m$ matches the size of~the circuit
        produced by~our algorithm~is
        \[s^*=\left(0,0,0,0,1,1,2,2,\dotsc,\frac{n}{2}-2, \frac{n}{2}-2\right).\]
        \begin{center}
            \begin{tikzpicture}
                \foreach \x [count=\n from 0] in {2, 1, ..., -6} {
                    \draw[l] (\x * \d, -1) -- (\x * \d, 0.75);
                    \node[below] at (\x * \d, -1) {$\n$};
                }

                \foreach \x in {-6, ..., 2} {
                    \node[dot] at (\x * \d, - 2 * \d) {};
                    \node[dot] at (\x * \d, - \d) {};
                }
                \node[dot] at (2 * \d, 0) {};
                \node[dot] at (2 * \d, \d) {};
            \end{tikzpicture}
        \end{center}
        In~this case, our method first spends $n/2$ gates to~pair the bits and then applies $n/2$ MDFA' blocks. The size of~the resulting circuit~is, up~to an~additive constant, $n/2+6\cdot n/2=3.5n$. This matches the upper bound $4.5n-2m$. Thus, to~improve the bound $4.5n-2m$ to~$4.5n-\beta m$, for $\beta > 2$, one needs to~find a~smaller circuit for this particular vector~$s^*$.
        And vice versa, by~proving a~lower bound $\size(\BA_n^{s^*}) \ge 3.5n-O(1)$,
        one would prove that $\alpha=2$.

        \item What is the smallest $\gamma$ such that
        \[\size(\BA_n^s) \le \gamma n-O(m)\]
        holds for every vector~$s$? In~this paper,
        we~proved that $\gamma \le 4.5$.
        Improving this seems to~be more challenging than just improving $2m$~to~$2.5m$ as~this would most probably require using new building blocks.

        \item Finally, note also that the upper bounds $5n-3m$ and $4.5n-2m$ hold for \emph{all} vectors~$s$.
        It~would be~interesting to~improve known upper and lower bounds for \emph{specific} vectors. Perhaps, one of~the most interesting such functions is~$\SUM_n$ (here, $s=(0,0,\dotsc,0)$). For~it, we~known an~upper bound $4.5n$ (originally proved
        by~Demenkov et al.~\cite{DBLP:journals/ipl/DemenkovKKY10}; also follows from our Theorem~\ref{theorem:main}) and a~lower bound $2.5n-O(1)$ due to~Stockmeyer~\cite{DBLP:journals/mst/Stockmeyer77}.
    \end{enumerate}

    \section*{Acknowledgments}
    We~thank the anonymous reviewers for many helpful comments.

    \bibliography{references}
\end{document}